\documentclass[11pt]{article}

\usepackage{times,fullpage,latexsym,amsfonts,amssymb,float,cite}

\usepackage{epic,eepic,epsf}
\usepackage{subfigure}
\usepackage{wrapfig}

\usepackage{algorithmic}
\usepackage{algorithm}

\usepackage{stmaryrd}

\newtheorem{theorem}{Theorem}
\newtheorem{lemma}{Lemma}

\newenvironment{proof}{\noindent\textbf{Proof Sketch:}}{\hfill$\Box$}
\newenvironment{proofsketch}{\noindent\textbf{Sketch of proof:}}{\hfill$\Box$}

\floatstyle{ruled}
\newfloat{algorithm}{Algorithm}{tbhp}[section]
\newtheorem{definition}{Definition}

\newcommand{\remove}[1]{}

\begin{document}

\begin{titlepage}

\renewcommand{\thefootnote}{\fnsymbol{footnote}}

\title{Practical Stabilizing Atomic Memory\\
{\rm\large (Extended Abstract)}}
\author{
\begin{tabular}{ccccc}
Noga Alon\footnotemark[1]& &
Hagit Attiya\footnotemark[2]& &
Shlomi Dolev\footnotemark[3]\\
Swan Dubois\footnotemark[4] & &
Maria Gradinariu\footnotemark[4] & &
S\'ebastien Tixeuil\footnotemark[4]
\end{tabular}}

\maketitle

\footnotetext[1]{Sackler School of Mathematics and Blavatnik School of Computer Science, Raymond and Beverly Sackler Faculty of Exact Sciences, Tel Aviv University, Tel Aviv, 69978, Israel. Email: {\tt nogaa@tau.ac.il}. Research supported
in part by an ERC advanced grant, by a USA-Israeli BSF grant, by the Israel Science Foundation.}

\footnotetext[2]{Department of Computer Science, Technion, 32000, Israel. Email: {\tt hagit@cs.technion.ac.il}.
Research supported in part by the {\em Israel Science Foundation}
(grant number 953/06).}

\footnotetext[3]{Contact author. Department of Computer Science,
Ben-Gurion University of the Negev, 
Beer-Sheva, 84105, Israel. Email: {\tt dolev@cs.bgu.ac.il}. The work started while this author was a visiting professor at LIP6.
Research supported in part by the ICT Programme of the European Union under contract
number FP7-215270 (FRONTS), Microsoft, Deutsche Telekom, US Air-Force
and Rita Altura Trust Chair in Computer Sciences.}

\footnotetext[4]{LIP6, Universite Pierre et Marie Curie, Paris 6/INRIA, 7606, France.}
 
\begin{abstract}
A self-stabilizing simulation of a single-writer multi-reader
atomic register is presented.
The simulation works in asynchronous message-passing systems,
and allows processes to crash, as long as at least a majority of them
remain working.
A key element in the simulation is a new combinatorial construction
of a bounded labeling scheme that can accommodate \emph{arbitrary}
labels, i.e., 
including those not generated by the scheme itself.
\end{abstract}

\thispagestyle{empty}\end{titlepage}

\renewcommand{\thefootnote}{\arabic{footnote}}

\section{Introduction}

Distributed systems have become an integral part of virtually all
computing systems, especially those of large scale.
These systems must provide high availability and reliability
in the presence of failures, which could be either permanent
or transient.

A core abstraction for many distributed algorithms simulates
shared memory~\cite{ABD95}; this abstraction allows to take
algorithms designed for shared memory, and port them to asynchronous
message-passing systems, even in the presence of failures.
There has been significant work on creating such simulations,
under various types of permanent failures, as well as on exploiting
this abstraction in order to derive algorithms for message-passing
systems.
(See a recent survey~\cite{A10}.)

All these works, however, only consider permanent failures,
neglecting to incorporate mechanisms for handling \emph{transient}
failures.
Such failures may result from incorrect initialization of the system,
or from temporary violations of the assumptions made by the system
designer, for example the assumption that a corrupted message
is always identified by an error detection code.
The ability to automatically resume normal operation following
transient failures, namely to be \emph{self-stabilizing}~\cite{D2K},
is an essential property that should be integrated
into the design and implementation of systems.

This paper presents the first practical self-stabilizing simulation of shared memory that tolerats crashes. 
Specifically,
 we propose a single-writer multi-reader (SWMR) atomic register in asynchronous message-passing systems where 
less than a majority of processors may crash. 
A single-writer multi-reader register is \emph{atomic} if each read operation returns the value of the most recent write 
operation happened before it or the value written by a concurent write.

The simulation is based on reads and writes to a (majority) quorum in a system with a 
fully connected graph topology\footnote{Note that the use of standard end-to-end schemes can be used to implement 
the quorum operation in the case of general communication graph.}. 
A key component of the simulation is a new bounded labeling scheme
that needs no initialization, as well as a method for using it
when communication links and processes are started at an arbitrary
state.

\paragraph*{Overview of our simulation.}

Attiya, Bar-Noy and Dolev~\cite{ABD95} presented the first simulation
of a SWMR atomic register in a message-passing system, supporting two
procedures, {\sf read} and {\sf write}, for accessing the register.
This simple simulation is based on a quorum approach:
In a {\sf write} operation, the writer makes sure that a quorum of
processors (consisting of a majority of the processors, in its
simplest variant) store its latest value.
In a {\sf read} operation, a reader contacts a quorum of processors,
and obtains the latest values they store for the register;
in order to ensure that other readers do not miss this value,
the reader also makes sure that a quorum stores its return value.

A key ingredient of this scheme is the ability to distinguish between
older and newer values of the register; this is achieved by attaching
a \emph{sequence number} to each register value.
In its simplest form, the sequence number is an unbounded integer,
which is increased whenever the writer generates a new value.
This solution could be appropriate for a an \emph{initialized} system,
which starts in a consistent configuration, in which all sequence
numbers are zero, and are only incremented by the writer
or forwarded as is by readers.
In this manner, a 64-bit sequence number will not wrap around for
a number of writes that is practically infinite,
certainly longer than the life-span of any reasonable system.

However, when there are transient failures in the system,
as is the case in the context of self-stabilization,
the simulation starts at an uninitialized state,
where sequence numbers are not necessarily all zero.
It is possible that, due to a transient failure, the sequence numbers
might hold the maximal value when the simulation starts running,
and thus, will wrap around very quickly.

Our solution is to partition the execution of the simulation into
\emph{epochs}, namely periods during which the sequence numbers
are supposed not to wrap around.
Whenever a ``corrupted'' sequence number is discovered,
a new epoch is started, overriding all previous epochs;
this repeats until no more corrupted sequence numbers are
hidden in the system, and the system stabilizes.
Ideally, in this steady state, after the system stabilizes,
it will remain in the same epoch
(at least until all sequence numbers wrap around,
which is unlikely to happen).

This raises, naturally, the question of how to label epochs.
The natural idea, of using integers, is bound to run into
the same problems as for the sequence numbers.
Instead, we capitalize on another idea from~\cite{ABD95},
of using a bounded labeling scheme for the epochs.
A \emph{bounded labeling scheme}~\cite{IL93,DS97} provides a
function for generating labels (in a bounded domain),
and guarantees that two labels can be compared to
determine the largest among them.

Existing labeling schemes assume that initially, labels have
specific initial values, and that new labels are introduced
only by means of the label generation function.
However, transient failures, of the kind the self-stabilizing
simulation must withstand, can create incomparable labels,
so it is impossible to tell which is the largest among
them or to pick a new label that is bigger than all of them.

To address this difficulty, we present a constructive bounded
labeling scheme that allows to define a label larger than
{\it any set} of labels, provided that its size is bounded.
We assume links have bounded capacity, and hence the
number of epochs initially hidden in the system is bounded.

The writer tracks the set of epochs it has seen recently;
whenever the writer discovers that its current epoch is not
the largest, or is incomparable to some existing epoch,
the writer generates a new epoch l that is larger than all
the epochs it has.
The number of bits required to represent a label depends
on $m$, the maximal size of the set, and it is in $O(m\log m)$.
We ensure that the size of the set is proportional to the
total capacity of the communication links, namely, $O(c n^2)$,
where $c$ is the bound on the capacity of each link, and hence,
each epoch requires $O((cn^2(\log n+\log c))$ bits.

It is possible to reduce this complexity, making $c$ essentially
constant, by employing a data-link protocol for communication
among the processors.

%

We show that, after a bounded number of {\sf write} operations,
the results of reads and writes can be totally casually ordered
in a manner that respects the read-time order of non-overlapping operations, 
so that the sequence of operations satisfies the
semantics of a SWMR register.
This holds until the sequence numbers wrap around, as can happen
in a realistic version of the unbounded ABD simulation.


\paragraph*{Related work.}
Self-stabilizing simulation of an atomic single-writer single-reader shared registers,
on a message-passing system, was presented in \cite{DIM97}.
This simulation does not address SWMR register.
Moreover, the simulation cannot withstand processor crashes.
More recent~\cite{DH01,JH09} papers focused on self-stabilizing simulation of shared registers using weaker shared registers. Self-stabilizing timestamps implementations using SWMR atomic registers were suggested in \cite{A03,DKS06}.
These implementations already assume the existence of a shared
memory, while, in contrast, we simulate a shared SWMR atomic register
using message passing.


\section{Preliminaries}
\label{sec:settings}

A \emph{message-passing system} consists of $n$ {\em processors},
$p_0,p_1,p_2,\ldots,p_{n-1}$,
connected by {\em communication links} through which messages are sent and received.
We assume that the underlying communication graph is completely
connected, namely, every pair of processors, $p_i$ and $p_j$,
have a communication link.

A processor is modeled by a state machine that executes
\emph{steps}.
In each step, the processor changes its state,
and executes a single communication operation,
which is either a {\em send} message operation
or a {\em receive} message operation.
The communication operation changes the state of an attached link,
in the natural manner.

The system {\em configuration} is a vector of $n$ states, a
state for each processors and $2(n^2-n)$ sets,
each bounded by a constant message capacity $c$.
A set $s_{ij}$ (rather than a queue, reflects the non-fifo nature)
for each directed edge ($i,j$) from a processor
$p_i$ to a processor $p_j$.
Note that in the scope of self-stabilization, where the system copes
with an arbitrary starting configuration, there is no deterministic data-link simulation 
that use bounded memory when the capacity of links is unbounded~\cite{DIM97}.

An {\it execution} is a sequence of configurations and steps,
$E=(C_1,a_1,C_2,a_2 \ldots)$ such that $C_i$, $i>1$, is obtained by
applying $a_{i-1}$ to $C_{i-1}$, where $a_{i-1}$ is a step of
a single processor, $p_j$, in the system.
Thus, the vector of states, except the state of $p_j$,
in $C_{i-1}$ and $C_{i}$ are identical.
In case the single communication operation in $a_{i-1}$ is a send
operation to $p_k$ then $s_{jk}$ in $C_i$ is a union of
$s_{jk}$ in $C_{i-1}$ with the message sent in $a_{i-1}$.
If the obtained union does not respect the
message bound $|s_{jk}|=c$ then an arbitrary
message in the obtained union is deleted.
The rest of the message sets are kept unchanged.
In case, the single communication operation in $a_{i-1}$ is a receive
operation of a (non null) message $m$,
then $m$ (must exist in $s_{kj}$ of $C_{i-1}$ and)
is removed from $s_{kj}$, all the rest of the sets are identical
in $C_{i-1}$ and $C_i$. A receive operation by $p_j$ from $p_k$ may
result in a null message even when the $s_{kj}$ is not empty,
thus allowing unbounded delay for any particular message.
Message losses are modeled by allowing
spontaneous message removals from the set.
An edge ($i,j$) is operational if a message sent infinitely often by
$p_i$ is received infinitely often by $p_j$.


For the simulation of a \emph{single writer multi-reader} (SWMR)
atomic register, we assume $p_0$ is the writer
and $p_1,p_2,\ldots,p_{n-1}$ are the readers.
$p_0$ has a {\sf write} procedure/operation and the readers have
{\sf read} procedure/operation. 
The sub-execution between
the step that starts a {\sf write}
procedure and the next step that ends
the {\sf write} procedure execution defines a {\it write period}.
Similarly, for a particular {\sf read} by processor $p_i$,
the sub-executions between the step that starts a {\sf read} procedure
by processors $p_i$ and the next step that ends the
{\sf read} procedure execution of $p_i$ defines a {\it read period}.

{\bf SWMR atomic register.}
A single-writer multi-reader atomic register supplies two operations: $read$ and $write$. 
An invocation of a $read$ or $write$  translates into a sequence of computation steps.
A sequence of invocations of $read$ and $write$ operations generates an execution in which the 
computation steps corresponding to different invocations are interleaved. 
An operation $op_1$ happens before an operation $op_2$ in this execution, if $op_1$ returns before $op_2$ is invoked. 
Two operations overlap if neither of them happens before the other. 
Each interleaved execution of an atomic register is required to be \emph{linearizable} \cite{lamport86}, that is, it must be equivalent 
to an execution in which the operations are executed sequentially, and the order of non-overlapping operations is preserved. 
The main difference between a regular register (a register that satisfies the property that every read retuns the value written 
by the most recent write or by a concurrent write) 
and an atomic register is the absence for the latter of the {\it new/old inversions}. 
Consider two consecutive\footnote{Two operations $op_1$ and $op_2$ are consecutives if $op_1$ is the most recent operation that {\it happens before} $op_2$.} 
reads  $r_1$, $r_2$ and two consecutive writes $w_1$, $w_2$ of a regular register such that  $r_1$ is concurrent with  both 
$w_1$ and $w_2$ and $r_2$ is concurrent only with $w_2$. 
The regularity property allows  $r_2$ to 
return the value writen by  $w_1$ and  $r_1$ to return the value writen by  
$w_2$. This phenomena is called the new/old inversion. 

An atomic register prevents in all executions the new/old inversions.

Formally, an atomic register verifies the following two properties:
\begin{itemize}
\item {\bf \emph{Regularity} property.} A $read$ operation
returns either the value written by the most recent
$write$ operation that happend before the $read$ or a value written by a concurrent $write$. 

\item {\bf \emph{No new old/inversions}} If a $read$ operation $r_1$ reads a value from a concurrent 
$write$ operation $w_2$ then no read operation that happens after $r_1$ reads a value from a 
write operation $w_1$ that happens before $w_2$. 
\end{itemize}

{\bf Practical stabilizing SWMR atomic register.}
A message passing system simulates a SWMR atomic register in a
practical stabilizing manner,
if any infinite execution starting in arbitrary configuration in which the writer writes infinitely often
has a sub-execution with a practically infinite number of write operations, in which the atomicity requirement holds.
A \emph{practically infinite execution} is an execution of at least
$2^k$ steps, for some large $k$;
for example, $k=64$ is big enough for any practical system.

\section{Overview of the Algorithm}
\label{sec:over}

\subsection{The Basic Quorum-Based Simulation}

We describe the basic simulation, 
which follows the quorum-based approach of~\cite{ABD95},
and ensures that our algorithm tolerates (crash) failures 
of less than a majority of the processors. Our simulation assumes the existance of an underlying 
stabilizing \emph{data-link}
protocol,\cite{dolev}, similar to the ping-pong mechanism used in~\cite{ABD95}.

The simulation relies on a set of \emph{read and write quorums},
each being 
a majority of processors.
The simulation specifies the {\sf write} and {\sf read} procedures,
in terms of {\sf QuorumRead} and {\sf QuorumWrite} operations.
The {\sf QuorumRead} procedure sends
a
request to every processor, for reading a certain local variable
of the processor; the procedure terminates with the obtained values,
after receiving answers from processors that form a quorum.
Similarly, the {\sf QuorumWrite} procedure sends a value to every
processor to be written to a certain local variable of the processor;
it terminates when acknowledgments from a quorum are received.
If a processor that is inside {\sf QuorumRead} or {\sf QuorumWrite}
keeps taking steps,
then the procedure terminates (possibly with arbitrary values).
Furthermore, if a processor starts {\sf QuorumRead} procedure execution, 
then the stabilizing data link~\cite{dolev} ensures
that a read of a value returns a value held by the read variable
some time during its period;
similarly, a {\sf QuorumWrite}($v$) procedure execution,
causes $v$ to be written to the variable during its period.

Each processor $p_i$ maintains a variable, $MaxSeq_i$, which is
meant to hold the ``largest'' sequence number the processor has read.
$p_i$ maintains in $v_i$ the value that $p_i$ knows for
the implemented register (which is associated with $MaxSeq_i$).

The {\sf write} procedure of a value $v$ starts with a {\sf QuorumRead} of the $MaxSeq_i$ variables;
upon receiving answers $l_1,l_2,\ldots$ from a quorum,
the writer picks a sequence number $l_m$ that is larger
than  $MaxSeq_0$ and $l_1,l_2,\ldots$ by one;
the writer assigns $l_m$ to $MaxSeq_0$ and calls {\sf QuorumWrite}
with the value $\langle{l_m, v}\rangle$.
Whenever a quorum member $p_i$ receives a {\sf QuorumWrite} request
$\langle l,v \rangle$ for which $l$ is larger than $MaxSeq_i$,
$p_i$ assigns $i$ to $MaxSeq_i$ and $v$ to $v_i$.

The {\sf read} procedure by $p_i$ starts with a {\sf QuorumRead}
of both the $MaxSeq_j$ and the (associated) $v_j$ variables.
When $p_i$ receives answers
$\langle{l_1, v_1}\rangle$, $\langle{l_2, v_2}\rangle \ldots$
from a quorum,
$p_i$ finds the largest label $l_m$ among $MaxSeq_i$,
and $l_1, l_2, \ldots$ and then
calls {\sf QuorumWrite} with the value $\langle l_m, v_m \rangle$.
This ensures that later {\sf read} operations will return this,
or a later, value of the register.
When {\sf QuorumWrite} terminates,
after a write quorum acknowledges,
$p_i$ assigns $l_m$ to $MaxSeq_i$ and $v_m$ to $v_i$ and
returns $v_m$ as the value read from the register.

Note that the {\sf QuorumRead} operation, beginning the write
procedure of $p_0$, helps to ensure that $MaxSeq_0$ holds the
maximal value, as the writer reads the biggest {\em accessible}
value (directly read by the writers, or propagated to
variables that are later read by the writer)
in the system during any write.

Let $g(C_1)$ be the number of distinct values greater
than $MaxSeq_0$ that exist in some configuration $C_1$.
Since all the processors, except the writer, only copy values
and since $p_0$ can only increment the value of $MaxSeq_0$
it holds for every $i\geq 1$ that
\[ g(C_i)\geq g(C_{i+1}) ~. \]
Furthermore,
\[ g(C_i) > g(C_{i+1}) ~, \]
whenever the writer discovers (when executing step $a_i$) a
value greater than $MaxSeq_0$.
Roughly speaking, the faster the writer discovers these values,
the earlier the system stabilizes.
If the writer does not discover such a value, then the (accessible)
portion of the system in which its values are repeatedly written, performs reads and writes correctly.

%

\subsection{Epochs}

As described in the introduction, it is possible that the sequence
numbers wrap around faster than planned, due to ``corrupted''
initial values.
When the writer discovers that this has happened,
it opens a new \emph{epoch},
thereby invalidating all sequence numbers from previous epochs.

Epochs are denoted with labels from a bounded domain,
using a \emph{bounded labeling scheme}.
Such a scheme provides a function to compute a new label,
which is ``larger'' than a given set of labels.

\begin{definition}
A \emph{labeling scheme} over a bounded domain ${\cal L}$,
provides an antisymmetric comparison predicate $\prec_b$ on ${\cal L}$
and a function $\mathbf{Next}(S)$ that returns a label
in ${\cal L}$, given some subset $S \subseteq {\cal L}$
of size at most $m$.
It is guaranteed that for every $L \in S$,
$L \prec_b \mathbf{Next}_b(S)$.
\end{definition}

Note that the labeling scheme 
~\cite{IL93}, used in the original atomic memory simulation~\cite{ABD95} does not 
cope with transient failures.
The next section describes a construction of a bounded labeling
scheme that can cope with badly initialized labels,
namely, that does not assume that labels were only generated
by using $\mathbf{Next}$.

Using this scheme, it is guaranteed that if the writer eventually
learns about all the epochs in the system, it will generate an
epoch greater than all of them.
After this point, any read that starts after a write of $v$ is
completed (written to a quorum) returns $v$ (or a later value),
since the writer will use increasing sequence numbers.

The eventual convergence of the labeling scheme depends on invoking
$\mathbf{Next}_b$ with a parameter $S$ that is a superset of
the epoches that are in the system.
Estimating this set is another challenge for the simulation.

We explain the intuition of this part of the simulation
through the following two-player \emph{guessing game},
between a \emph{finder}, representing the writer,
and a \emph{hider}, representing an adversary controlling the system.
\begin{itemize}
\item[--]
The hider maintains a set of labels ${\cal H}$,
whose size is at most $m$ (a parameter that will be chosen later).
\item[--]
The finder does not know ${\cal H}$, but it would like to generate
a label greater than all labels in ${\cal H}$.
\item[--]
The finder generates a label $L$ and if ${\cal H}$
contains a label $L'$,
such that it does not hold that $L' \prec_b L$,
then the hider exposes $L'$ to the finder.
\item[--]
In this case, the hider may choose to add $L$ to ${\cal H}$,
however, it must ensure that the size of ${\cal H}$ remains
smaller than $m$ (by removing another label).
(The finder is unaware of the hiders decision.)
\item[--]
If the hider does not expose a new label $L'$ from ${\cal H}$
the finder wins this iteration and continues to use $L$.
\end{itemize}

The finder uses the following strategy.
It maintains a fifo queue of $2m$ labels,
meant to track the most recent labels.
The queue starts with arbitrary values,
and during the course of the game,
it holds up to $m$ recent labels produced by the finder,
that turned out to be overruled by existing labels
(provided by the hider).
The queue also holds up to $m$ labels that were revealed
to overrule these labels.

Before the finder chooses a new label,
it enqueues its previously chosen label
and the label received from the hider in response.
Enqueuing a label that appears in the queue pushes the
label to the head of the queue;
if the bound on the size of the queue is reached,
then the oldest label in the queue is dequeued. 
This semantics of enqueue is
used throughout the paper.

The finder choose the next label by applying $\mathbf{Next}$,
using as parameter the $2m$ labels in the queue.
Intuitively, the queue eventually contains a superset of ${\cal H}$,
and the finder generates a label greater than all the current
labels of the hider.

\begin{lemma}
All the labels of the hider are smaller than one of the first $m+1$
labels chosen by the finder.
\end{lemma}

\begin{proofsketch}
A simple induction shows that when the finder chooses the $i$th new
label $i>0$, the $2i$ items in the front of the queue consist
of the first $i$ labels generated by the finder,
and the first $i$ labels revealed by the hider. 

Note that a response cannot expose a label that has been introduced or
previously exposed in the game since the finder always choose a label
greater than all labels in the queue, in particular these $2i$ labels. Thus, if the finder does not win when introducing the $m$th label,
all the $m$ labels that the hider had when the game started were 
exposed and therefore, stored in the queue of the finder together 
with all the recent $m$ labels introduced by the finder,
before the $m+1$st label is chosen.
Therefore, the $m+1$st label is larger than every label held by the 
hider, and the finder wins.
\end{proofsketch}

\subsection{Timestamps}

The complete simulation tags each value written with a
\emph{timestamp}---a pair $\langle{l,i}\rangle$,
where $l$ is an epoch chosen from a bounded domain ${\cal L}$
and $i$ is a sequence number
(an integer smaller than some bound $r$).

\section{A Bounded Labeling Scheme with Uninitialized Values}
\label{sec:boundedlabel}

Let $k > 1$ be an integer, and let $K = k^2+1$.
We consider the set $X=\{1,2,..,K\}$ and let ${\cal L}$
(the set of labels) be the set of all ordered pairs $(s,A)$
where $s \in X$ is called in the sequel the {\em sting} of $X$,
and $A \subseteq X$ has size $k$ and is called in the sequel
{\em Antistings} of $X$.
It follows that $|{\cal L}|={K \choose k} K = k^{(1+o(1))k}$.

The comparison operator $\prec_b$ among the bounded labels is
defined to be: [[i and j replaced]]
\[
(s_j,A_j) \prec_b (s_i,A_i) \equiv (s_j \in A_i) \wedge (s_i \not \in A_j)
\]

Note that this operator is antisymmetric by definition,
yet may not be defined for every pair $(s_i,A_i)$ and $(s_j,A_j)$
in ${\cal L}$ (e.g., $s_j \in A_i$ and $s_i \in A_j$).

We define now a function to compute, given a subset $S$ of at most $k$ labels of ${\cal L}$, a new label which is greater (with respect to $\prec_b$) than every label of $S$.
This function, called $\mathbf{Next}_b$ (see Figure \ref{f:UBC}) is
as follows.
Given a subset of $k$ label $(s_1,A_1)$, $(s_2,A_2)$, $\ldots$, $(s_k,A_k)$, we construct a label $(s_i,A_i)$ which satisfies:
\begin{itemize}
\item[--]
$s_i$ is an element of $X$ that is not in the union $A_1 \cup A_2 \cup \ldots \cup A_k$ (as the size of each $A_s$ is $k$, the size of the union is  at most $k^2$,
and since $X$ is of size $k^2+1$ such an $s_i$ always exists).
\item[--]
$A$ is a subset of size $k$ of $X$ containing all values $(s_1, s_2,\ldots, s_k)$ (if they are not pairwise distinct,
add arbitrary elements of $X$ to get a set of size exactly $k$).
\end{itemize}

\begin{figure*}[htb!]
\begin{scriptsize}
\centering
\begin{tabular}{|p{2.9in}||p{2.9in}|}
\hline&\\
\begin{minipage}[t]{2.5in}
\centering
{\it\bf Next$_b$}
\begin{tabbing}
X: \= d \= d \= d \= d \= d \= d \= d \= d \= \kill

\textbf{input:} $S =(s_1,A_1), (s_2,A_2), \ldots, (s_k,A_k)$:
                set of labels\\

\textbf{output:} $(s,A)$: label\\

\textbf{function:} For any $\emptyset\neq S\subseteq X$,
 $pick(S)$ returns arbitrary\\ (later defined for particular cases) element of $S$\\

1: \>$A := \{s_1\}\cup\{s_2\}\cup \ldots \cup \{s_k\}$\\
2: \>\emph{while} $|A|\neq k$\\
3: \> \> $A := A\cup\{pick(X\setminus A)\}$\\
4: \> \> $s := pick\left(X\setminus(A \cup A_1\cup A_2\cup \ldots \cup A_k)\right)$\\
5: \>\emph{return} $(s,A)$\\
\end{tabbing}
\end{minipage}
&

\begin{minipage}[t]{2.5in}
\centering
{\it\bf Next$_e$}
\begin{tabbing}
X: \= d \= d \= d \= d \= d \= d \= d \= d \= \kill

\textbf{input:} $S$: set of $k$ timestamps \\

\textbf{output:} $(l,i)$: timestamp\\

1: \>\emph{if} $\exists (l_0,j_0) \in S$ such that\\
   \> \> $\forall (l,j) \in S, (l,j)\neq (l_0,j_0),
                  (l,j) \prec_e (l_0,j_0) \wedge j_0 < r$ \\
2: \> \>\emph{then return} $(l_0, j_0+1)$\\
3: \>\emph{else return} $(\mathbf{Next}_b(\tilde{S}),0)$\\
\end{tabbing}

\end{minipage}\\[1ex]
\hline
\end{tabular}
\caption{Next$_b$ and Next$_e$.
$\tilde{S}$ denotes the set of labels appearing in $S$.}
\label{f:UBC} \end{scriptsize}
\vspace*{-0.5cm}
\end{figure*}

\begin{lemma}
Given a subset $S$ of $k$ labels from ${\cal L}$,
$(s_i,A_i)= \mathbf{Next}_b(S)$ satisfies:
\[
\forall (s_j,A_j)\in S, (s_j,A_j) \prec_b (s_i,A_i)
\]
\end{lemma}

\begin{proof}
Let $(s_j,A_j)$ be an element of $S$.
By construction, $s_j\in A_i$ and $s_i \notin A_j$,
and the result follows from the definition of $\prec_b$.
\end{proof}

Note also that it is simple to compute $A_i$ and $s_i$
given a set $S$ with $k$ labels,
and can be done in time linear in the total length of the labels
given, i.e., in $O(k^2)$ time.
Since the number of labels $|{\cal L}|$ is $k^{(1+o(1))k}$, we have
that $k$ is $\frac{(1+o(1)) \log |{\cal L}|}{\log \log |{\cal L}|}$.

\paragraph*{Timestamps.}
\label{sec:enumerablelabel}

A \emph{timestamp} is a pair $(l,i)$ where $l$ is a bounded epoch,
and $i$ is an integer (sequence number),
ranging from $0$ to a fixed bound $r \geq 1$.

The $\mathbf{Next}_e$ operator compares between two timestamps,
and is described in Figure~\ref{f:UBC}.
Note that in line 3 of the code we use $\tilde{S}$ for the set of
labels (with sequence numbers removed) that appear in $S$.
The comparison operator $\prec_{e}$ for timestamps is:
\[
(x,i)\prec_e (y,j) \equiv x\prec_b y \vee ( x = y \wedge i < j )
\]

In the sequel, we use $\prec_b$ to compare timestamps,
according to their epochs only.

\section{Putting the Pieces Together}
\label{sec:gg}

Each processor $p_i$ holds, in $MaxTS_i$,
two fields $\langle ml_i, cl_i \rangle$,
where $ml_i$ is the timestamp associated with the last write of a value to the variable $v_i$
and $cl_i$ is a \emph{canceling timestamp} possibly empty ($\bot$),
which is not smaller than $MaxTS_i.ml$ in the $\prec_b$ order.
The canceling field is used to let the writer (finder in the game) to know an evidence.
A timestamp ($l,i$) is an evidence for timestamp
($l',j)$ if and only if $l \not\prec_b l'$.
When the writer faces an evidence it changes the current epoch.

The pseudo code for the read and write procedures appears in
Figure~\ref{f:wr}.
Note that in lines 2 and 9 of the {\sf write} procedure,
a label is enqueued if and only if it is not equal to
$MaxTS_0$.
Note further, that $Next_e$ in line 4 of the writer, first
tries to increment the sequence number of the label stored
in $MaxTS_0$ and if the sequence number already equals to the
upper bound $r$ then $p_0$ enqueues the value of $MaxTS_0$
and use the updated $epochs$ queue to choose a new value
for $MaxTS_0$, which is a new epoch Next$_b$($epochs$)
and sequence number $0$.

\begin{figure*}[htb!]
\begin{scriptsize}
\centering
\begin{tabular}{|p{2.9in}||p{2.9in}|}
\hline&\\
\begin{minipage}[t]{2.5in}
\centering
{\it\bf write$_0$($v$)}
\begin{tabbing}
XX: \= d \= d \= d \= d \= d \= d \= d \= d \= \kill

1: \>$\langle \langle ml_1,cl_1 \rangle, v_1 \rangle, \langle \langle ml_2,cl_2 \rangle, v_2 \rangle, \cdots :=${\sf QuorumRead}\\
2: \>$\forall i$, {\em if} $ml_i \neq MaxTS_0$ {\em then} enqueue($epochs,ml_i$)\\
3: \>$\forall i$, {\em if} $cl_i \neq MaxTS_0$ {\em then} enqueue($epochs,cl_i$)\\
4: \>{\em if} $\forall$ $l \in epochs$ $l \preceq_e MaxTS_0$ {\em then}\\
5: \>\>$MaxTS_0:=\langle Next_e$($MaxTS_0$,$epochs$), $\bot \rangle$\\
6: \>{\em else}\\
7: \>\>enqueue($epochs,MaxTS_0$)\\
8: \>\>$MaxTS_0:=\langle (Next_b(epochs),0), \bot \rangle$\\
9: \>{\sf QuorumWrite}($\langle MaxTS_0.ml, v\rangle$)\\
\\
Upon a request of {\sf QuorumWrite} $\langle l, v \rangle$\\
10: \>{\em if} $l \neq MaxTS_0$ {\em then} enqueue($epochs,l$)\\
\end{tabbing}
\end{minipage}
&
\begin{minipage}[t]{2.5in}
\centering
{\it\bf read}
\begin{tabbing}
XX: \= d \= d \= d \= d \= d \= d \= d \= d \= \kill
1: \>$\langle \langle ml_1,cl_1 \rangle, v_1 \rangle, \langle \langle ml_2,cl_2 \rangle, v_2 \rangle, \cdots :=${\sf QuorumRead}\\
2: \>\emph{if} $\exists m$ such that $cl_m=\bot$ {\em and}\\
3: \> \> ($\forall$ $i\neq m$ $ml_i \prec_e ml_m$ and $cl_i \prec_e ml_m$) then\\
4: \> \> \>{\sf QuorumWrite}$\langle ml_m, v_m \rangle$\\
5: \> \> \>return($v_m$)\\
6: \>\emph{else} return($\bot$)\\
\\
\\
\\
Upon a request of {\sf QuorumWrite} $\langle l, v \rangle$\\
7: \>{\em if} $MaxTS_i.ml \prec_e l$ and $MaxTS_i.cl \prec_e l$ {\em then}\\
8: \> \> $MaxTS_i:=\langle l, \bot \rangle$\\
9: \> \> $v_i:=v$\\
10: \> else {\em if} $l \not\prec_b MaxTS_i.ml$ and $MaxTS_i.ml \neq l$ \\  \hspace{1cm}{\em then} $MaxTS_i.cl:=l$\\
\end{tabbing}
\end{minipage}\\[1ex]
\hline
\end{tabular}
\caption{write($v$) and read.}
\label{f:wr}
\end{scriptsize}
\vspace*{-0.0cm}
\end{figure*}

The {\sf write} procedure of a value $v$ starts with a {\sf QuorumRead} of the $MaxTS_i$ variables, and upon receiving
answers $l_1,l_2,\ldots$ from a quorum, 
the writer $p_0$ enqueues to the {\em epochs} queue the epochs of the received $ml$ and non-$\bot$ $cl$ values, 
which are not equal to $MaxTS_0$ (lines 1-3).
The writer then computes $MaxTS_0$ to be the Next$_e$ timestamp, 
namely if the epoch of $MaxTS_0$ is the largest in the $epochs$ queue and the sequence number of $MaxTS_0$ less than $r$, 
then $p_0$ increments the sequence number of $MaxTS_0$ by one,
leaving the epoch of $MaxTS_0$ unchanged (lines 4-5). 
Otherwise, it is necessary to change the epoch:
$p_0$ enqueues $MaxTS_0$ to the {\em epochs} queue and applies 
$Next_b$ to obtain an epoch greater than all the ones in the $epochs$ queue;
it assigns to $MaxTS_0$ the timestamp made of this epoch 
and a zero sequence number (lines 7-8).
Finally, $p_0$ executes the {\sf QuorumWrite} procedure with 
$\langle MaxTS_0, v \rangle$ (line 9).

Whenever the writer $p_0$ receives (as a quorum member) a 
{\sf QuorumWrite} request containing an epoch that is not 
equal to $MaxTS_0$, $p_0$ enqueues the received label in {\em epochs} queue (line 10).

The {\sf read} procedure executed by a reader $p_i$ starts with a 
{\sf QuorumRead} of the $MaxTS_j$ and the (associated) $v_j$ variables (line 1).
When $p_i$ receives answers 
$\langle \langle ml_1, cl_1 \rangle, v_1 \rangle$, 
$\langle \langle ml_2, cl_2 \rangle, v_2 \rangle \ldots$ 
from a quorum,
$p_i$ tries to find a maximal timestamp $ml_m$ according to the
$\prec_e$ operator from among 
$ml_i$, $cl_i$, $ml_1$, $cl_1$, $ml_2$, $cl_2$ $\ldots$.
If $p_i$ finds such maximal timestamp $ml_m$, 
then $p_i$ executes the {\sf QuorumWrite} procedure with 
$\langle ml_m, v_m \rangle$.
Once the {\sf QuorumWrite} terminates 
(the members of a quorum acknowledged) 
$p_i$ assigns $MaxTS_i:=\langle ml_m, \bot \rangle$, 
and $v_i:=v_m$ and returns $v_m$ as the
value read from the register (lines 2-5).
Otherwise, in case no such maximal value $ml_m$ exists, 
the read is aborted (line 6).

When a quorum member $p_i$ receives a {\sf QuorumWrite} request $\langle l,v \rangle$, 
it checks whether both $MaxTS_i.ml \prec_b l$
and  $MaxTS_i.cl \prec_b l$.
If this is the case, then $p_i$ assigns 
$MaxTS_i:=\langle l, \bot \rangle$ 
and $v_i:=v$ (lines 7-9).
Otherwise, $p_i$ checks whether $l \not\prec_b MaxTS_i.ml$
and if so assigns $MaxTS_i.cl:=l$ (line 10).
Note that $\bot \prec_b l$, for any $l$. 

Note that we assume the existance of an underlying  data-link protocol that emulates FIFO links 
over a non-FIFO communication environment.
 In the following we assume that the data-link protocol also helps in 
repeatedly transmit the value of $MaxTS$ from one processor to another.
In case the $MaxTS_i.cl$ of a processor $p_i$ is $\bot$ and $p_i$
receives from a neighbor $p_j$ a $MaxTS_j$ such that
$MaxTS_j.ml \not\prec_b MaxTS_i.ml$ then $p_i$ assigns
$MaxTS_i.cl:=MaxTS_j.ml$, otherwise, when
$MaxTS_j.cl \not\prec_b MaxTS_i.ml$ then $p_i$ assigns
$MaxTS_i.cl:=MaxTS_j.cl$. Note also that the writer will enqueue every diffused value different from $MaxTS_0$. The code is identical to  line 9 in the writer code.

\subsection{Outline of Correctness Proof}
\label{sec:together}

The correctness of the simulation is implied by the game and 
our previous observations, which we can now summarize, recapping the arguments
explained in the the description of the individual components.

In the simulation, the finder/writer may introduce new epochs
even when the hider does not introduce an evidence.
We consider a timestamp ($l,i$) to be an evidence for timestamp
($l',j)$ if and only if $l \not\prec_b l'$.
Using large enough bound $r$ on the sequence number
(e.g., a 64-bit number),
we ensure that either there is a practically infinite
execution in which the finder/writer introduces new timestamps with no
epoch change, and therefore with growing sequence numbers,
and well-defined timestamp ordering,
or a new epoch is frequently introduced due to
the exposure of hidden unknown epochs.
The last case follows the winning strategy described for the game.

The sequence numbers allow the writer to introduce many (practically infinite) timestamps without storing all of them,
as their epoch is identical.
The sequence numbers are a simple extension of the bounded epochs
just as a least significant digit of a counter;
allowing the queues to be proportional to the bounded
number of the labels in the system.
Thus, either the writer introduces an epoch greater than
any one in the system,
and hence will use this epoch to essentially implement a
register for a practically unbounded period,
or the readers never introduce some existing bigger epoch 
letting the writer increment the sequence number infinitely often.
Note that if the game continues, while the finder is aware of
(a superset including) all existing epochs, and introduces a greater epoch, there is a practically
infinite execution before a new epoch is introduced.

In the scope of simulating a SWMR atomic register, following the
first write of a timestamp greater than any other timestamp in the
system, with a sequence number 0, to a majority quorum, any read
in a practically infinite execution,
will return the last timestamp that has been written to a quorum.
In particular, if a reader finds a timestamp introduced by
the writer that is larger than all other timestamps but not
yet completely written to a majority quorum, the reader assists
in completing the write to a majority quorum before returning
the read value.

The memory may stop operate while the set
of timestamps does not include a timestamp greater than the
rest. That is, read operations may be
repeatedly aborted until the writer writes
new timestamps. Moreover, a slow reader may store
a timestamp unknown to the rest (and in particular to
the writer) and eventually introduce the timestamp to the
rest.
In the first case the convergence of the
system is postponed till the writer is aware
of a superset of the existing timestamps. In the second case 
the system operate correctly,
implementing read and write operations, until
the timestamp unknown to the rest is introduced.

\begin{theorem}
The algorithm eventually reaches a period in which it
simulates a SWMR atomic register, for a number of operations
that is linear in $r$.
\end{theorem}

Each {\sf read} or {\sf write} operation requires $O(n)$ messages.
The size of the messages is linear in the size of a timestamp,
namely the sum of the size of the epoch and $\log r$.
The size of an epoch is $O( m log m)$ where $m$ is the
size of the {\em epochs} queue,
namely, $O(c n^2)$,
where $c$ is the capacity of a communication link.

Note that the size of the {\em epochs} queue, and with it, the size
of an epoch, is proportional to the number of labels that
can be stored in a system configuration.
Reducing the link capacity will reduce the number of labels
that can be ``hidden'' in the communication links.
This can be achieved by using a stabilizing \emph{data-link}
protocol,\cite{dolev}, 
in a manner similar to the ping-pong mechanism used in~\cite{ABD95}.

\section{Conclusion}

We have presented a self-stabilizing simulation of a single-writer
multi-reader atomic register, in an asynchronous message-passing
system in which at most half the processors may crash.

Given our simulation, it is possible to realize a self-stabilizing
\emph{replicated state machines} \cite{L98}.
The self-stabilizing consensus algorithms presented in \cite{DKS06}
uses SWMR registers, and our simulation allows to port them to
message-passing systems.
More generally, our simulation allows the application of any
self-stabilizing algorithm that is designed using SWMR registers
to work in a message-passing system,
where at most half the processors may crash.

Our work leaves open many interesting directions for future research.
The most interesting one is to find a stabilizing simulation, 
which will operate correctly even after sequence numbers wrap around, without an additional convergence period. 
This seems to mandate a more carefully way to track epochs [[]],
perhaps by incorporating a self-stabilizing analogue of the
\emph{viability} construction~\cite{ABD95}.

\paragraph*{Acknowledgments.}
We thank Ronen Kat and Eli Gafni for helpful discussions.

\newpage
\section*{Anexes}

\begin{lemma}
\label{lem:hidden-revealed}
Every execution has an infinite suffix where every hidden timestamp is eventually revealed 
to the writter or stays hidden forever (not revealed neither to the writter nor to a reader) .
\end{lemma}

\begin{proof} 
Consider an execution where a timestamp is not revealed directly 
to the writter but to some clean reader (a reader with canceling setted to $\bot$).
The other cases are trivial.
Let $l$ be the timestamp and $i$ be the reader.
Following the description of the code piggy-backed by the data-link 
then $i$ compares $MaxTS_i.ml$ with $l$. If $l \not \prec_e MaxTS_i$ then 
$MaxTS_i.cl$ is setted to $l$.  Then, either the writter contacts the reader via a QuorumRead 
and gets the canceling field or the reader is contacted by another clean reader and 
the canceling is propagated. Eventually, the writter will get the 
canceled timestamp and enqueues it.  
\end{proof}

\begin{lemma}
\label{lem:maxtimestamp}
Each infinite execution has an infinite suffix where every QuorumRead invocation by a reader
returns a maximum clean timestamp.
\end{lemma}

\begin{proof}
We prove in the following that  the prefix where QuorumRead invocation by a reader 
returns either  canceled timestamps or timestamps that do not have a clean maximum is finite.
The proof is by construction.
Every write operation  invokes a QuorumWrite with a clean timestamp that is greater than any timestamp the writter is aware of. Therefore, every QuorumRead invoked after the QuorumWrite invocation captures this value. According to Lemma \ref{lem:hidden-revealed} every hidden timestamp is eventually either revealed to the writter and enqueued or stays hidden. Since the number of hidden values is bounded, the writter enqueues these values in a finite time. Consider the execution after the writer enqueues the last hidden value. The next write operation produces a timestamp that is greater than any timestamp that will be ever revealed in the execution and any QuorumRead invoked after the execution of this write will get this timestamp. 
\end{proof}

\begin{lemma}
\label{lemma:abort}
Each execution of the system has an infinite suffix where reads do
not abort.
\end{lemma}

\begin{proof}
According to Lemma \label{lem:maxtimestamp} every execution has an infinite suffix where each QuorumRead invocation returns a maximum clean timestamp. It follows that for every read invocation,  
the conditions in lines 2 and 3 (reader's code) are satisfied and the value returned by the read is not $\bot$.
\end{proof}

\begin{lemma}
\label{lemma:regularity}
Any execution of the system has an infinite suffix that satisfies the
regularity property.
\end{lemma}

\begin{proof}
Let $e$ be an infinite execution of the system.
Following Lemma \ref{lemma:abort} and Lemma \ref{lem:hidden-revealed}, $e$ contains an infinite suffix, 
$e^\prime$, where any read returns a not abort value and any 
write includes in its decision set all the labels in the system.
Assume there is a process $p$ such that it read invocations 
allways return an obsolete value. That is, the value returned by the
read is either a hidden value or a value corresponding to a previous
write but not the most recent. Let $r$ be such a read.
In $e^\prime$, $r$ returns the output value with the maximum timestamp over
the set of labels returned by QuorumRead. 
Let $w_1$ and $w_2$ be two write operations such that $w_1$ happens
before $w_2$ and $r$. Since $w_1$ happens before $r$ then the label
computed by $w_1$ is written in at least a majority of processes via a 
QuorumWrite and is greater than any label in
the system. When $r$ starts invoking QuorumRead two cases may appear:
(1)$w_2$ didn't modify the value written by $w_1$ and didn't start 
its promotion via QuorumWrite or (2)$w_2$ executes QuorumWrite but didn't finish its execution.
In the first case, $w_1$'s MaxTS is the largest in the
system. When $r$ invokes the QuorumRead it gets $w_1$'s MaxTS value
(otherwise $w_1$ is not terminated) and returns it.
Hence, $r$ cannot return a value older than the one written by $w_1$.
In the second case, some processes contacted in the QuorumRead may
send the $w_1$'s MaxTS, other processes the $w_2$'s MaxTS. Since the 
MaxTS computation at the writter is sequential then $w_2$'s MaxTS 
is greater than $w_1$'s MaxTS. Then following lines 2 and 3 in the reader code, 
$r$ should return $w_2$'s MaxTS. 
Hence, $r$ will return the last written value.
\end{proof}

\begin{lemma}
Any execution of the system has an infinite suffix that satisfies the
no new/old inversion property.
\end{lemma}

\begin{proof}
Let $e$ be an execution of the system.
Following Lemmas \ref{lemma:abort} and \ref{lemma:regularity}, $e$ has
an infinite suffix, $e^\prime$, that satisfies the regularity property and in which any read
invocation does not return abort.
In the following we prove that $e^\prime$ does not violate the new/old inversion property. 
Consider two write operations $w_1$ and $w_2$ in $e^\prime$ such that 
$w_1$ happens before $w_2$. Consider also two read operations $r_1$ and $r_2$ such that $r_1$ 
happens before $r_2$ and $w_1$ happens before $r_1$\footnote{Following the 
transivity of the relation happens before, $w_1$ also happens before $r_2$.}.
Assume $r_1$ and $r_2$ are concurrent with $w_2$. Assume a new/old inversion 
happens and $r_1$ returns the value written by $w_2$. Let denote the MaxTS of this value with $l_2$. 
Assume also $r_2$ retuns the value written by $w_1$ whch MaxTS is $l_1$. 
Since $r_1$ happens before $r_2$ then before the start of $r_2$, $r_1$ executes
the following actions: it modifies its MaxTS to $l_2$, it also 
executes QuorumWrite in order to inform the system of its new value. 
Since QuorumWrite retuns before the $r_1$ finishes then 
$l_2$ is already adopted by at least a majority of processes. 
That is, since $l_2 \succ_e l_1$ ($w_1$ happens before $w_2$), then $l_2$ replaces 
$l_1$ in at least a 
majority of processes.

We assumed $r_2$ returns $l_1$. Since $r_1$ happens before $r_2$ then 
$r_2$ starts its QuorumRead after $r_1$ returned so after $r_1$ completed 
its QuorumWrite operation.
This implies that $l_2$ is the label adopted by at least a majority of 
processes and at least one process in this majority will respond while $r_2$ invokes its QuorumRead. 
That is, the $r_2$ collects at least one label $l_2$ and since 
$l_2 \succ_e l_1$, $r_2$ should return this value.
This contradicts the assumption $r_2$ retuns $l_1$.
It follows that $e^\prime$ verifies the no new/old inversion property.  
 \end{proof}
\end{document}